\documentclass[conference]{IEEEtran}

\usepackage[letterpaper,left=0.680in,right=0.680in,bottom=0.990in,top=0.71in]{geometry}
\usepackage{ifpdf}
\ifCLASSINFOpdf
  \usepackage[pdftex]{graphicx}
  \graphicspath{{./results_v3/figures/}}
  \DeclareGraphicsExtensions{.pdf,.png,.jpeg}
\else
  \usepackage[dvips]{graphicx}
  \DeclareGraphicsExtensions{.eps}
\fi
\usepackage{amsmath,amssymb,amsfonts}
\usepackage{amsthm}
\usepackage{booktabs}
\usepackage{array}
\usepackage{cite}
\usepackage{url}
\usepackage{balance}
\usepackage[caption=false,font=footnotesize]{subfig}
\usepackage{xcolor}
\usepackage{tikz}
\usetikzlibrary{arrows.meta,positioning,calc,fit,backgrounds}

\definecolor{hlcol}{RGB}{22,110,62}
\tikzset{
  blk/.style={draw=black!70,rounded corners=1.6pt,align=center,font=\scriptsize,
              text width=20mm,minimum height=8mm,inner sep=2pt},
  hi/.style={blk,draw=hlcol,line width=0.9pt,fill=hlcol!7},
  off/.style={draw=black!55,dashed,rounded corners=1.6pt,align=center,
              font=\scriptsize,inner sep=3pt,fill=black!4},
  ar/.style={-{Latex[length=1.8mm,width=1.3mm]},line width=0.5pt,black!75},
  arhi/.style={-{Latex[length=1.8mm,width=1.3mm]},line width=0.8pt,hlcol}
}

\newtheorem{proposition}{Proposition}
\newtheorem{lemma}{Lemma}
\newtheorem{remark}{Remark}

\begin{document}

\title{Inverse Maxwell-Based Wall-Aware OFDM-ISAC \\ for Slow-Moving Target Sensing}

\author{
  \IEEEauthorblockN{Tri~Nhu~Do, Yosefine~Triwidyastuti, and Gunes~Karabulut~Kurt}
  \IEEEauthorblockA{Department of Electrical Engineering, Polytechnique Montr\'eal, Montr\'eal, QC, Canada\\
  \{tri-nhu.do, yosefine-2.triwidyastuti, gunes.kurt\}@polymtl.ca}
}

\maketitle

\begin{abstract}
In this paper, we study integrated sensing and communication (ISAC) for short-range indoor orthogonal frequency-division multiplexing (OFDM) systems in which the sensing path crosses a building wall. The target is a slow-moving user equipment behind the wall, and its echo is embedded in the wall reflection and static indoor clutter. Because the wall adds excess propagation length, attenuation, and internal reflections, a conventional delay transform reports an apparent range. We therefore formulate physical-range estimation as an inverse Maxwell problem for the known wall: the range operator is the distorted-Born Jacobian of the discretized one-dimensional Helmholtz equation, assembled offline from the calibrated wall and applied online as a per-snapshot range inversion, so that the range-Doppler map is indexed by physical rather than apparent range. We then prove that this operator carries a structural limitation: when the two-way wall factor has constant magnitude and linear phase, the wall-aware range image is identically the free-space image on a translated grid, for every ridge level, taper, and noise realization. A free-space branch translated by the same excess length is therefore a required comparison. Sensing-only numerical results for a representative layered wall, with every branch on one range grid and a common post-FFT CSI model, show that the proposed inverse removes the range bias that wall-unaware processing cannot, and that a scalar-corrected free-space inverse tracks it closely, the two differing only marginally in empirical range RMSE and detection probability across the tested sweep.
\end{abstract}


\section{Introduction}

Orthogonal frequency-division multiplexing (OFDM) is attractive for sensing because subcarrier and symbol indices encode delay and Doppler, and because a colocated receiver that knows the transmitted symbols can form a coherent range-Doppler image \cite{DaiCOMST2026}; the same channel-state information (CSI) structure underlies OFDM channel estimation \cite{Nguyen_TGCN_2026}, is reused by integrated sensing and communication (ISAC) for data transmission \cite{LuoTWC2024}, and lets WLAN sensing detect human presence from beacon-frame CSI \cite{MathWorksWLANSensingDL}. When the propagation path crosses a building wall, the delay axis is no longer the physical range axis: the wall adds electrical length, attenuation, and internal reflections, so a conventional delay transform reports an apparent range with a range bias.

Through-wall radar has long treated that range bias as a focusing error: a wave travels more slowly inside a dielectric wall, so a receiver that assumes free-space propagation reports the excess delay as excess range, while the slab also attenuates the path, ripples the spectrum through internal reflections, and refracts the wave \cite{LiaoTAES2025}. Wideband beamformers and synthetic-aperture processors compensate the excess delay and the refraction in order to restore a free-space image \cite{AhmadTAES2005,DehmollaianTGRS2008}. Linearized inverse scattering about a known background is the corresponding model-based device \cite{ChewTMI1990}. High-order Maxwell discretizations with embedded interface conditions \cite{LawJCP2026} make a wall-aware Jacobian of the OFDM CSI look like a ready-made operator; the numerics below use a finite-difference discretization of the 1-D reduction instead. The research question is therefore not whether such an operator reduces range bias, but what the reduction can be attributed to. A calibrated wall \cite{YaoTIM2025} supplies two distinct pieces of information: a scalar range bias $\delta_w$, which any receiver can subtract from its delay estimate, and a frequency-dependent transfer function, which only a model-based operator can exploit. If the scalar accounts for the entire gain, the improvement reduces to subtracting a known constant.

In this paper, we consider short-range indoor OFDM sensing in which a wall separates the transceiver from the target. The user equipment (UE) behind the wall is the sensing target; it moves slowly, and its echo competes with the reflection from the wall and with returns from static indoor scatterers. The data-bearing OFDM waveform serves both the downlink and monostatic sensing. We formulate physical-range estimation as an inverse problem and propose a background-calibrated inverse Maxwell operator, assembled offline from the known wall and applied online during sensing. Major contributions of the present study include the following:
\begin{itemize}
  \item We propose an inverse Maxwell method for OFDM ISAC that senses a slow-moving target behind a known wall: the range operator is the distorted-Born Jacobian of the discretized Helmholtz equation of the calibrated wall, assembled offline and applied online per snapshot, so that the range-Doppler map is indexed by physical rather than apparent range. We show, and verify numerically, that it reduces to the squared wall transmission times a free-space steering kernel with a constant independent of frequency and range.
  \item We prove that, for a constant-magnitude, linear-phase two-way wall factor, the wall-aware estimator is a relabeling of the free-space estimator shifted by $\delta_w$. A free-space branch translated by $\delta_w$ is therefore a required comparison.
  \item Numerical results on a representative $L$-layer dielectric stack, with all branches on one range grid and a common post-FFT CSI model, show that the residual ripple and phase curvature of the wall factor leave both the range RMSE and the detection probability essentially unchanged; the scalar-corrected free-space inverse likewise departs from the Maxwell inverse only marginally over the same sweep.
\end{itemize}

\section{System Model}

\begin{figure}[t]
\centering
\includegraphics[width=0.80\columnwidth]{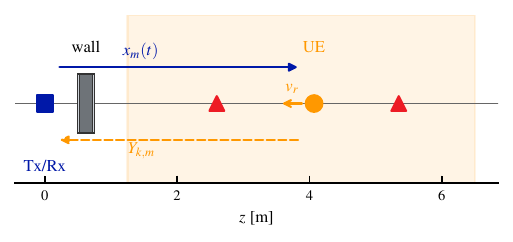}
\caption{Sensing geometry. Square: transceiver at $z=0$. Slab: known layered wall. Triangles: static clutter. Circle: UE at range $r_0$, approaching at $v_r$.}
\label{fig:geometry}
\end{figure}

We consider a phase-coherent monostatic OFDM transceiver located at $z=0$ in a short-range indoor setting, as illustrated in Fig.~\ref{fig:geometry}. The OFDM waveform carries data for a downlink and is reused for sensing; this work evaluates only the latter. The UE behind the wall at physical range $r_0$ is the sensing target; its two-way echo, collected at the colocated transceiver \cite{DaiCOMST2026}, is the sensing path. A single spatial channel supports neither angle nor shape estimation, so the demonstration is a 1-D normal-incidence reduction.

A known wall occupies $[z_w,z_w+d_w]$ and consists of $L$ homogeneous, isotropic, nonmagnetic layers with thicknesses $d_i$, relative permittivities $\varepsilon_{r,i}$, and conductivities $\sigma_i$, so that $d_w=\sum_{i=1}^{L}d_i$. The UE lies at physical range $r_0>z_w+d_w$ with radial velocity $v_r$, positive when approaching, and is modeled as an effective point target for sensing. The data-generating wall is an $L$-layer dielectric stack (Table~\ref{tab:config}); the inverse model is a single homogeneous slab whose one-way travel time is calibrated to that stack.

\subsection{Transmit OFDM-ISAC Signal}

The transceiver uses $K$ subcarriers at spacing $\Delta f$, occupied bandwidth $B=K\Delta f$ (the span of $\{f_k\}$ is $(K-1)\Delta f$), and subcarrier frequencies $f_k=f_c+(k-K/2)\Delta f$ for $k=0,\ldots,K-1$; we write $\tilde f_k=f_k-f_c$. Range resolution uses the DFT convention $\Delta R=c/(2B)$. Let $S_{k,m}$ denote the known unit-modulus quadrature phase-shift keying (QPSK) data symbol on subcarrier $k$ of snapshot $m$. The baseband OFDM symbol is
\begin{equation}
x_m(t)=\sum_{k=0}^{K-1}S_{k,m}\,e^{j2\pi\tilde f_k t},\qquad t\in[-T_{\rm cp},T_u],
\label{eq:xm}
\end{equation}
where $T_u=1/\Delta f$ and $T_{\rm cp}$ is the cyclic prefix (CP); the symbol duration is $T_{\rm sym}=T_u+T_{\rm cp}$. Sensing snapshots are taken at slow time $t_m=mT_r$, with $T_r\neq T_{\rm sym}$ in general.

\subsection{Channel Modeling: Medium, Wall, and Clutter}

We adopt the $e^{+j\omega t}$ convention, for which forward propagation is $e^{-jkz}$. For normal incidence on an isotropic nonmagnetic medium, a transverse scalar field $u$ satisfies the Helmholtz equation \cite{Chew1995Waves}
\begin{equation}
\frac{d^2u}{dz^2}+k_0^2\varepsilon_{r,c}(z,\omega)u=q,\quad k_0=\omega/c,\ q=j\omega\mu_0 J_{s,x},
\label{eq:helmholtz}
\end{equation}
with complex relative permittivity $\varepsilon_{r,c}=\varepsilon_r-j\sigma/(\omega\varepsilon_0)$ and outgoing waves in the exterior half spaces.

The wall response is obtained from the exact multilayer transfer matrix. With refractive indices $n_i=\sqrt{\varepsilon_{r,c,i}}$ on the principal branch, so that $\Im\{n_i\}\leq 0$, the interface and propagation matrices are
\begin{align}
\mathbf I_{i,i+1}&=\frac12
\begin{bmatrix}
1+n_{i+1}/n_i & 1-n_{i+1}/n_i\\
1-n_{i+1}/n_i & 1+n_{i+1}/n_i
\end{bmatrix},\\
\mathbf P_i&=\operatorname{diag}\!\bigl(e^{+jk_0 n_i d_i},e^{-jk_0 n_i d_i}\bigr).
\end{align}
With $n_0=n_{L+1}=1$ and $\mathbf T=\mathbf I_{0,1}\prod_{i=1}^{L}\bigl(\mathbf P_i\mathbf I_{i,i+1}\bigr)$, the stack reflection and transmission are $r_w=T_{21}/T_{11}$ and $t_w=1/T_{11}$. Since the UE is parameterized by its physical range, the free-space delay over the wall extent is de-embedded: the one-way factor and the monostatic two-way wall factor are
\begin{equation}
t_{w,{\rm rel}}(f)=t_w(f)e^{+jk_0 d_w},\qquad
g_w(f)=\bigl[t_{w,{\rm rel}}(f)\bigr]^{2},
\label{eq:gw}
\end{equation}
where the square expresses reciprocity of the two one-way transits at normal incidence. By construction $g_w\equiv 1$ for an all-air stack of any thickness.

The wall reflection and static scatterers at fixed ranges form a calibrated background CSI $H_{{\rm bg},k}$, which is not a sensing parameter.

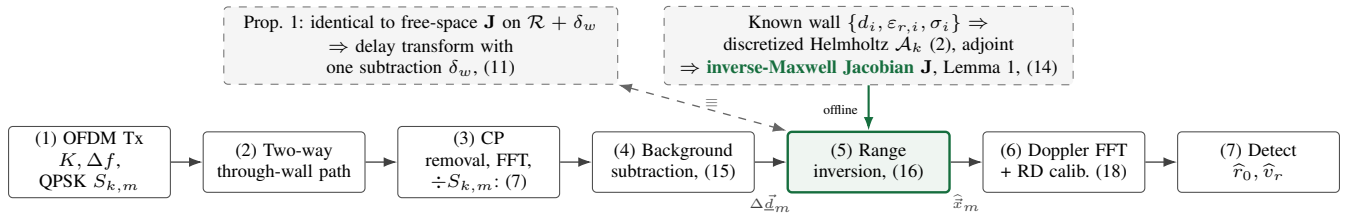
\begin{figure*}[t]
\centering
\begin{tikzpicture}[node distance=4.2mm]
\node[blk] (tx)   {(1) OFDM Tx\\$K,\Delta f$, QPSK $S_{k,m}$};
\node[blk,right=of tx]  (prop) {(2) Two-way\\through-wall path};
\node[blk,right=of prop] (rx)  {(3) CP removal, FFT,\\{}$\div S_{k,m}$: \eqref{eq:Ykm}};
\node[blk,right=of rx]  (bg)   {(4) Background\\subtraction, \eqref{eq:diffcsi}};
\node[hi, right=of bg]  (inv)  {(5) Range\\inversion, \eqref{eq:ridgesvd}};
\node[blk,right=of inv] (dop)  {(6) Doppler FFT\\{}+ RD calib. \eqref{eq:rdnorm}};
\node[blk,right=of dop] (det)  {(7) Detect\\$\widehat r_0,\widehat v_r$};

\foreach \a/\b in {tx/prop,prop/rx,rx/bg,dop/det}{\draw[ar] (\a)--(\b);}
\draw[ar] (bg)--(inv) node[midway,font=\tiny,inner sep=0pt,below=4.2mm] {$\Delta\underline{\vec d}_m$};
\draw[ar] (inv)--(dop) node[midway,font=\tiny,inner sep=0pt,below=4.2mm] {$\widehat{\vec x}_m$};

\node[off,above=6mm of inv,text width=52mm] (wall)
  {Known wall $\{d_i,\varepsilon_{r,i},\sigma_i\}$ $\Rightarrow$ discretized Helmholtz $\mathcal A_k$ \eqref{eq:helmholtz}, adjoint\\[1pt]
   $\Rightarrow$ \textbf{\color{hlcol}inverse-Maxwell Jacobian} $\mathbf J$, Lemma~\ref{lem:jacobian}, \eqref{eq:jwall}};
\draw[arhi] (wall)--(inv);
\node[font=\tiny,anchor=east,inner sep=1pt] at ($(wall.south)+(-1mm,-3.2mm)$) {offline};

\node[off,left=6mm of wall,text width=50mm] (eqv)
  {Prop.~\ref{prop:equiv}: identical to free-space $\mathbf J$ on $\mathcal R+\delta_w$\\[1pt]
   $\Rightarrow$ delay transform with one subtraction $\delta_w$, \eqref{eq:deltaw}};
\draw[{Latex[length=1.8mm,width=1.3mm]}-{Latex[length=1.8mm,width=1.3mm]},dashed,
        black!60,line width=0.5pt] (eqv.south east)--(inv.north west)
        node[pos=0.5,above right,font=\tiny,inner sep=0.5pt] {$\equiv$};
\end{tikzpicture}
\caption{Processing chain of the wall-aware OFDM ISAC receiver. Steps (1)--(4) and (6)--(7) are conventional OFDM radar; the numerics represent (1)--(3) by the post-FFT CSI model \eqref{eq:Ykm} and run (4)--(5) with one-dimensional range detection in every Monte Carlo trial, while steps (6)--(7) are exercised on the single representative CPI of Fig.~\ref{fig:rdmaps}. The inverse Maxwell model enters only in the two shaded blocks; Proposition~\ref{prop:equiv} states that step~(5) is reproducible by a free-space inverse on the grid translated by $\delta_w$ (dashed box at left).}
\label{fig:pipeline}
\end{figure*}

\subsection{Received Signal at the UE and at the Sensing Receiver}

The one-way OFDM field at the UE, after transit of the wall, is the downlink communication observation,
\begin{equation}
u_{\rm UE}(f)=t_{w,{\rm rel}}(f)\,e^{-jk_0 r_0}\,X(f),
\label{eq:utgt}
\end{equation}
where $X(f)$ is the frequency-domain transmit symbol. It is recorded only to identify the two-way sensing path; bit-error rate and rate are not reported.

The two-way sensing echo of the same UE is received at the colocated transceiver. After CP removal and the subcarrier FFT, the frequency-domain sensing observation on snapshot $m$ is
\begin{equation}
\underline{Y}_{k,m}=H_{k,m}S_{k,m}+\underline{N}_{k,m},
\label{eq:Ykm}
\end{equation}
where $\underline{N}_{k,m}\sim\mathcal{CN}(0,\sigma_{\rm rx}^2)$ is independent across subcarriers and snapshots, and the CSI is the superposition of clutter and the dynamic UE,
\begin{equation}
H_{k,m}=H_{{\rm bg},k}+\alpha\,g_w(f_k)\,e^{-j4\pi f_k r_m/c},
\label{eq:Hkm}
\end{equation}
where $r_m=r_0-v_r(m-M/2)T_r$ is the UE range on snapshot $m$, $r_0$ its range at the centre of the coherent processing interval (CPI) of $M$ snapshots, and the carrier part of the phase, $e^{+j2\pi\nu mT_r}$ with $\nu=2f_c v_r/c$, is the monostatic Doppler. The band-dependent remainder is the range walk over the CPI and the full-band wideband Doppler spread $2|v_r|B/c$; both are kept in the data and neglected by the fixed range operator of Section~\ref{sec:proposed}. Since $|S_{k,m}|=1$, the sensing CSI $\underline{Y}_{k,m}/S_{k,m}=H_{k,m}+\underline{N}_{k,m}/S_{k,m}$ has noise of the same distribution. Equation~\eqref{eq:Ykm} is the output of steps (1)--(3) of Fig.~\ref{fig:pipeline} for a causal channel supported on $[0,T_{\rm cp}]$ and time invariant within one symbol: for the tone sum \eqref{eq:xm}, the received signal on $[0,T_u]$ equals $\sum_k S_{k,m}H(f_k)e^{j2\pi\tilde f_k t}$, and its $K$ samples return $H(f_k)S_{k,m}$ through the DFT. The numerics therefore start from \eqref{eq:Ykm} rather than executing steps (1)--(3), with a direct-path round trip $\tau_{\max}=2(r_{\mathrm{ROI,max}}+\delta_w)/c<T_{\rm cp}$ for $\delta_w$ of \eqref{eq:deltaw}; the internal reflections of the wall are not of finite support, and Section~\ref{sec:results} verifies the sampled OFDM identity for a finite-delay test channel. We therefore adopt \eqref{eq:Ykm} as an ideal synchronized post-FFT observation model, neglecting intercarrier interference and any channel energy outside the CP; all reported results are conditional on it.

\subsection{Detection and Estimation Task}

The sensing task is to detect the through-wall UE in a range region of interest (ROI) and to estimate its physical range $r_0$ and radial velocity $v_r$. A conventional chain applies a delay transform on the subcarriers of the CSI and a Doppler FFT across snapshots, which estimates an \emph{apparent} free-space range $r^{\rm app}=c\tau/2$. The proposed chain, presented in Section~\ref{sec:proposed}, replaces the delay transform by a wall-aware inverse and then applies the same Doppler FFT. Under negligible migration, the reduced unknown at range cell $r_p$ is the complex range-cell coefficient
\begin{equation}
x_{p,m}=a_p e^{+j2\pi\nu m T_r},
\label{eq:xpm}
\end{equation}
so that one fixed linear range operator applies at every $m$ and preserves the phase required by the slow-time FFT; a literal material parameter would not carry that phase.

\section{Problem Formulation}
\label{sec:problem}

Let $\Delta\underline{\vec{d}}_m\in\mathbb{C}^K$ denote the background-subtracted CSI on snapshot $m$, obtained by subtracting one averaged no-UE reference from the CSI of \eqref{eq:Ykm}. A conventional delay processor estimates
\begin{equation}
\widehat{r}^{\rm app}
=\arg\max_{r_q^{\rm app}}
\Bigl\lvert
\sum_{k=0}^{K-1}w_k\,\Delta\underline{D}_{k,m}\,e^{+j4\pi\tilde f_k r_q^{\rm app}/c}
\Bigr\rvert,
\label{eq:delayest}
\end{equation}
where $w_k$ is a Hann taper. Because of the wall, $\widehat{r}^{\rm app}$ differs from $r_0$ by a systematic \emph{range bias}, a deterministic offset that remains at infinite SNR, of approximately the geometric excess
\begin{equation}
\delta_w=\sum_{i=1}^{L}\bigl(\sqrt{\varepsilon_{r,i}}-1\bigr)d_i,
\label{eq:deltaw}
\end{equation}
which for a single slab is $d_w(\sqrt{\varepsilon_r}-1)$. Equation~\eqref{eq:deltaw} is the range equivalent of the two-way excess delay, using the real lossless indices; it is not the band group delay of $g_w$.

The objective is to recover the physical range from the same CSI. Let $\mathbf{J}\in\mathbb{C}^{K\times P}$ be a range operator on a physical-range grid $\mathcal{R}=\{r_p\}_{p=1}^P$. The range image on snapshot $m$ is obtained from the ridge problem
\begin{equation}
\widehat{\vec{x}}_m
=\arg\min_{\vec{x}}
\bigl\|\mathbf{W}_f(\Delta\underline{\vec{d}}_m-\mathbf{J}\vec{x})\bigr\|_2^2
+\lambda\|\vec{x}\|_2^2,
\label{eq:opt}
\end{equation}
where $\mathbf{W}_f=\operatorname{diag}(w_0,\ldots,w_{K-1})$. With $\mathbf{W}$ the linear processor that maps $\Delta\underline{\vec{d}}_m$ to $\widehat{\vec{x}}_m$ and $G_p=\sum_k|W_{p,k}|^2$ its row noise gain, the range estimate is $\widehat{r}=\arg\max_{r_p\in\mathcal{R}}|\widehat{x}_{p,m}|^2/G_p$; the same noise-gain-normalized statistic drives detection in Section~\ref{sec:results}. The Doppler estimate is obtained from the slow-time FFT of $\widehat{\mathbf{X}}=[\widehat{\vec{x}}_0,\ldots,\widehat{\vec{x}}_{M-1}]$.

The choice of $\mathbf{J}$ is the key: a free-space steering matrix reproduces the range bias of \eqref{eq:delayest}, whereas a Maxwell-derived Jacobian of the known wall embeds $g_w(f)$ so that the peak is indexed by physical range. Neither \eqref{eq:opt} nor $\mathbf{J}$ separates, by itself, the scalar $\delta_w$ from the frequency-dependent content of $g_w$; that separation is the subject of Section~\ref{sec:proposed}.

\section{Proposed Wall-Aware Inverse Processing}
\label{sec:proposed}

Fig.~\ref{fig:pipeline} summarizes the receiver. Only step~(5) distinguishes the proposed chain from a conventional OFDM radar; the wall enters it through the Jacobian derived next.

\subsection{Background-Calibrated Distorted-Born Jacobian}

Let $\chi=\varepsilon_{r,c}-\varepsilon_{\rm bg}$ be the contrast and let $\mathcal{A}_k(\chi)$ be a discretized Helmholtz operator \eqref{eq:helmholtz} at subcarrier $k$, let $\vec{q}_k$ be the source, and let $\vec{c}^T$ be the receiver sampling functional, so that $\mathcal{A}_k(\chi)\vec{u}_k=\vec{q}_k$ and $d_k(\chi)=\vec{c}^T\vec{u}_k$. The known background $\chi_{\rm bg}$ is the wall alone: static clutter is additive in \eqref{eq:Hkm} and is removed by the subtraction of \eqref{eq:diffcsi}. For a small dynamic contrast $\delta\chi_m$,
\begin{equation}
\mathcal{A}_{{\rm bg},k}\,\delta\vec{u}_{k,m}
\approx-\delta\mathcal{A}_{k,m}\,\vec{u}_{{\rm bg},k},
\end{equation}
and the differential observation is linear in $x_{p,m}$.

\begin{lemma}[Distorted-Born Jacobian and its layered reduction]
\label{lem:jacobian}
For a unit point source and receiver at $z=0$ and a nodal scalar contrast with point weights $\Delta z_p$,
$J_{{\rm bg},k,p}=-k_{0,k}^2 g_{k,p}u_{{\rm bg},k,p}\Delta z_p$
with $\mathcal{A}_{{\rm bg},k}^{T}\vec{g}_k=\vec{c}$. For the 1-D layered background of Section~II,
\begin{equation}
J_{{\rm bg},k,p}=\frac{\Delta z_p}{4}\,g_w(f_k)\exp\!\bigl(-j4\pi f_k r_p/c\bigr),
\label{eq:jwall}
\end{equation}
whose constant depends on neither $k$ nor $p$.
\end{lemma}
\begin{proof}
Differentiating $d_k=\vec{c}^T\mathcal{A}_k^{-1}\vec{q}_k$ gives $\delta d_k=-\vec{g}_k^T\delta\mathcal{A}_k\vec{u}_{{\rm bg},k}$ with $\mathcal{A}_{{\rm bg},k}^T\vec{g}_k=\vec{c}$ and $\delta\mathcal{A}_k=k_0^2\operatorname{diag}(\delta\chi)$. The transpose, not the Hermitian transpose, appears because $\vec{c}^T\mathcal{A}^{-1}$ is the reciprocal forward derivative. The outgoing Green's function of \eqref{eq:helmholtz} in 1-D is $G(z,z')=\tfrac{j}{2k_0}e^{-jk_0|z-z'|}$, so behind the wall $u_{\rm bg}(r_p)=\tfrac{j}{2k_0}t_{w,{\rm rel}}(f)e^{-jk_0 r_p}$, and reciprocity gives $g(r_p)=u_{\rm bg}(r_p)$. Hence $J=-k_0^2u_{\rm bg}(r_p)^2\Delta z_p$, and $-k_0^2\bigl(\tfrac{j}{2k_0}\bigr)^2=\tfrac14$ yields \eqref{eq:jwall}. In one dimension there is no geometric spreading.
\end{proof}

The Born factor $-k_0^2$ is cancelled exactly by the two Green's functions, so no frequency window is dropped. Concretely, with $h$ the node spacing, $e_s$ the transceiver node and $i_p$ the node at $r_p$, the finite-dimensional model is $\mathcal{A}_k(\vec{a})=\mathcal{A}_{{\rm bg},k}+k_{0,k}^2\sum_p a_p(\Delta z_p/h)e_{i_p}e_{i_p}^{T}$ with $\vec{q}_k=e_s/h$: the target enters as an additive point potential on the fixed background, and the modified-wavenumber rule that assembles $\mathcal{A}_{{\rm bg},k}$ is not itself differentiated. Complex symmetry gives $\vec{g}_k=h\vec{u}_{{\rm bg},k}$ and $\partial d_k/\partial a_p|_{\vec{a}=0}=-k_{0,k}^2\Delta z_p u_{{\rm bg},k,i_p}^2$, the exact derivative of that model; \eqref{eq:jwall} is its continuum limit and the finite-$h$ gap is quantified below. The weights $\Delta z_p$ normalize the point scatterers on $\mathcal{R}$; they are not exact integrals over finite-width material cells. The reduced operator is \eqref{eq:jwall} with unit-norm columns and the carrier phase $e^{-j4\pi f_c r_p/c}$ absorbed into $x_{p,m}$, i.e., $J_{k,p}\propto g_w(f_k)e^{-j4\pi\tilde f_k r_p/c}$. The numerics assemble $\mathcal{A}_k$ as a second-order finite-difference operator on a $0.1$~mm grid (70\,501 nodes) with the modified wavenumber $2(1-\cos k_0 n h)/h^2$, which makes plane waves exact in each homogeneous region, and exact discrete outgoing conditions; $\mathcal{A}_k$ is complex symmetric, so one banded solve per subcarrier gives $\vec{u}_{\rm bg}$ and $\vec{g}$. The assembled Jacobian divided by \eqref{eq:jwall} is independent of $p$ to numerical precision, its constant is flat across the band, and the reduced operator agrees with the closed-form kernel. The free-space baseline is the same assembly with $\varepsilon_r\equiv1$; data generation uses the exact layered transfer matrix, the inverse the travel-time-calibrated slab.

Collecting snapshots, the measured differential CSI is
\begin{equation}
\Delta\underline{\mathbf{D}}
=\mathbf{J}_{\rm true}\mathbf{X}
+\underline{\mathbf{N}}_{\rm target}
-\underline{\mathbf{N}}_{{\rm bg,ref}}\vec{1}_M^T,
\label{eq:diffcsi}
\end{equation}
where $\mathbf{J}_{\rm true}$ is the true layered linear map, $\mathbf{X}$ collects $x_{p,m}$, and the background reference is the average of $N_{\rm bg}$ no-UE snapshots, subtracted at every $m$.

\subsection{Ridge Inverse and Range-Doppler Processing}

With $\widetilde{\mathbf{J}}=\mathbf{W}_f\mathbf{J}=\mathbf{U}\boldsymbol{\Sigma}\mathbf{V}^H$, the solution of \eqref{eq:opt} is
\begin{equation}
\widehat{\vec{x}}_m
=\mathbf{V}\operatorname{diag}\!\Bigl(\frac{\sigma_i}{\sigma_i^2+\lambda}\Bigr)
\mathbf{U}^H\mathbf{W}_f\Delta\underline{\vec{d}}_m
\equiv\mathbf{W}_\lambda\Delta\underline{\vec{d}}_m.
\label{eq:ridgesvd}
\end{equation}
The ridge level $\lambda=10^{-2}\sigma_{\max}^2$, with $\sigma_{\max}$ the larger of the two operators' leading singular values, is shared by the wall-aware and free-space inverses, so regularization of the ill-conditioned oversampled dictionary cannot confound the comparison. Both branches then take the same slow-time Hann-weighted Doppler FFT,
\begin{equation}
Z_{p,\ell}=\sum_{m=0}^{M-1}w_m^{\rm s}\widehat{X}_{p,m}e^{-j2\pi\ell m/M},
\label{eq:doppfft}
\end{equation}
so that a target factor $e^{+j2\pi\nu m T_r}$ peaks at $\nu_\ell=\ell/(MT_r)$ and $v_\ell=c\nu_\ell/(2f_c)$. The FFT acts on complex $\widehat{X}_{p,m}$; taking $|\widehat{X}_{p,m}|$ first would destroy the carrier Doppler phase.

Because the same averaged reference is subtracted at every $m$, its error is not white in slow time: with the receiver noise white and the $N_{\rm bg}$ reference snapshots independent of the target-present ones, the noise in $Z_{p,\ell}$ has variance $G_p\bigl[\sigma_{\rm rx}^2\sum_m|w_m^{\rm s}|^2+\sigma_{\rm rx}^2 N_{\rm bg}^{-1}|V_\ell|^2\bigr]$ for a processor with row gain $G_p$ and window spectrum $V_\ell=\sum_m w_m^{\rm s}e^{-j2\pi\ell m/M}$. The calibrated range-Doppler (RD) power
\begin{equation}
P_{p,\ell}
=\frac{|Z_{p,\ell}|^2}
{G_p\bigl[\sigma_{\rm rx}^2\sum_m|w_m^{\rm s}|^2+\sigma_{\rm rx}^2 N_{\rm bg}^{-1}|V_\ell|^2\bigr]}
\label{eq:rdnorm}
\end{equation}
has unit mean under the noise-only hypothesis for every $(p,\ell)$; the second term is concentrated in the mainlobe (bins $0,\pm1$) of the symmetric Hann window $w^{\rm s}_m=\tfrac12[1-\cos(2\pi m/(M-1))]$, small but nonzero elsewhere, the simulated target at $0.02$~m/s sits in bin $2$, and the calibration is shared by every method.

\begin{table}[t]
\caption{Simulation parameters.}
\label{tab:config}
\centering
\footnotesize
\setlength{\tabcolsep}{2.5pt}
\renewcommand{\arraystretch}{1.06}
\begin{tabular}{lcl}
\toprule
Parameter & Value & Derived quantity\\
\midrule
$f_c$ & 5.8 GHz & $\lambda_c=51.7$ mm\\
$K,\Delta f$ & $128$, 4 MHz & $B=512.0$ MHz\\
$T_u/T_{\rm cp}/T_{\rm sym}$ & $0.25/0.10/0.35~\mu$s & $\tau_{\max}=45.35$ ns\\
$T_r$, $M$ & 10 ms, 256 & $T_{\rm CPI}=MT_r=2.56$ s\\
Nominal DFT spacing & -- & $0.2928$ m / $0.01010$ m/s\\
Wall (truth) & $d_1=d_3=0.025$ m & $\varepsilon_{r,1,3}=2.5$, $\sigma=0.010$\\
 & $d_2=0.200$ m & $\varepsilon_{r,2}=5.5$, $\sigma=0.030$\\
Wall (inverse) & $d_w=0.250$ m & $\varepsilon_r=4.8$, $\sigma=0.025$\\
Wall position & $z_w=0.50$ m & 2-way loss (truth) $9.56$ dB\\
Target & $r_0=4.07$ m & $v_r=0.02$ m/s\\
Range grid & $[1.25,6.50]$ m & step $5$ mm ($1051$ cells)\\
Helmholtz grid & $h=0.1$ mm & $70\,501$ nodes\\
Ridge / reference & $10^{-2}\sigma_{\max}^2$ & $N_{\rm bg}=32$\\
\bottomrule
\end{tabular}
\end{table}

\subsection{Scalar-Delay Equivalence}

Write $g_w(f)=A(f)e^{j\varphi(f)}$ and split the unwrapped phase over the band into its least-squares linear part and a remainder,
\begin{equation}
\varphi(f)=\varphi_0-\frac{4\pi\tilde f\delta_w^{\rm grp}}{c}+\varphi_{\rm r}(\tilde f),
\label{eq:phaseexp}
\end{equation}
where $\delta_w^{\rm grp}$ is $-c/(4\pi)$ times the least-squares slope over the $K$ subcarriers, so that $\varphi_{\rm r}$ has zero mean and zero least-squares slope. The geometric length \eqref{eq:deltaw} approximates $\delta_w^{\rm grp}$ but is not equal to it: for the nominal slab the two differ by $1.6$~mm and for the truth stack by $6.8$~mm, well inside a range resolution cell.

\begin{proposition}[Scalar-delay equivalence]
\label{prop:equiv}
Suppose $g_w(f_k)=A e^{j\varphi_0}e^{-j4\pi\tilde f_k\delta_w/c}$ for constants $A>0$, $\varphi_0\in\mathbb{R}$, and $\delta_w\in\mathbb{R}$. Assume $\mathcal{R}$ and $\mathcal{R}+\delta_w$ are congruent (same spacing and cardinality); $\delta_w$ need not be an integer multiple of the grid step. Let $\mathbf{W}_{\lambda,{\rm wall}}$ be built from \eqref{eq:jwall} on $\mathcal{R}$, and let $\mathbf{W}_{\lambda,{\rm free}}$ be built with $g_w\equiv 1$ on $\mathcal{R}+\delta_w$. Then, for every $\lambda>0$, every taper $\mathbf{W}_f$, and every $\Delta\underline{\vec{d}}$,
$\mathbf{W}_{\lambda,{\rm wall}}=e^{-j\varphi_0}\mathbf{W}_{\lambda,{\rm free}}$.
Hence $|\widehat{x}^{\rm wall}(r_p)|=|\widehat{x}^{\rm free}(r_p+\delta_w)|$, the range estimates satisfy $\widehat{r}_{\rm wall}=\widehat{r}_{\rm free}-\delta_w$ identically in the data, and the detection statistics coincide.
\end{proposition}
\begin{proof}
The $p$th unnormalized wall-aware column is $A e^{j\varphi_0}$ times the free-space column at range $r_p+\delta_w$. Unit-norm column scaling cancels $A$, so $\mathbf{J}_{\rm wall}=e^{j\varphi_0}\mathbf{J}_{\rm free}'$. The ridge solution of \eqref{eq:opt} is $\mathbf{W}_\lambda=(\widetilde{\mathbf{J}}^H\widetilde{\mathbf{J}}+\lambda\mathbf{I})^{-1}\widetilde{\mathbf{J}}^H\mathbf{W}_f$; the Gram matrix is invariant under the unit-modulus scalar and $\widetilde{\mathbf{J}}^H$ acquires $e^{-j\varphi_0}$, which is the stated identity. A unit-modulus scalar changes neither $|\widehat{x}_p|$ nor $G_p$, and \eqref{eq:doppfft} acts on $m$ only.
\end{proof}

\begin{remark}[Required comparison]
\label{rem:comparison}
Under Proposition~\ref{prop:equiv}, the wall-aware inverse carries no information about the wall beyond $\delta_w$. A comparison that omits a $\delta_w$-corrected branch cannot attribute the observed range-bias reduction to wave physics.
\end{remark}

Any contribution of such an operator must come from the residuals that the hypothesis excludes: the amplitude taper $A(f)/A(f_c)$ and the phase curvature $\varphi_{\rm r}$. For the nominal slab these are a $1.65$~dB ripple and an $11^\circ$ phase residual with about two Fabry--P\'erot periods across $B$; for the truth stack, whose graded index steps roughly halve each interface reflection, they are appreciably smaller.

\begin{figure}[tb]
\centering
\includegraphics[width=0.82\columnwidth]{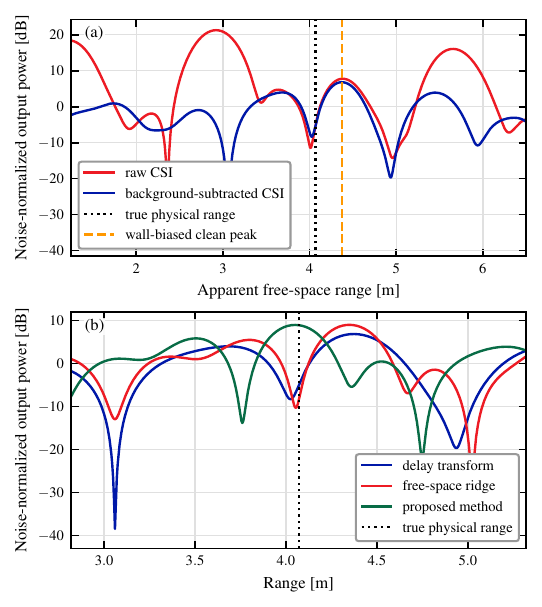}
\caption{Single-snapshot range profiles at $-10$~dB differential SNR, the ordinate being the noise-normalized output power $|Z|^2/\operatorname{var}(Z\mid\mathcal{H}_0)$, whose linear mean is $1+$ output SNR. (a)~Raw and background-subtracted CSI on the apparent axis. (b)~On the common grid, the delay transform and free-space ridge peak beyond $r_0$; the wall-aware inverse does not.}
\label{fig:profiles}
\end{figure}

\begin{figure}[tb]
\centering
\includegraphics[width=0.90\columnwidth]{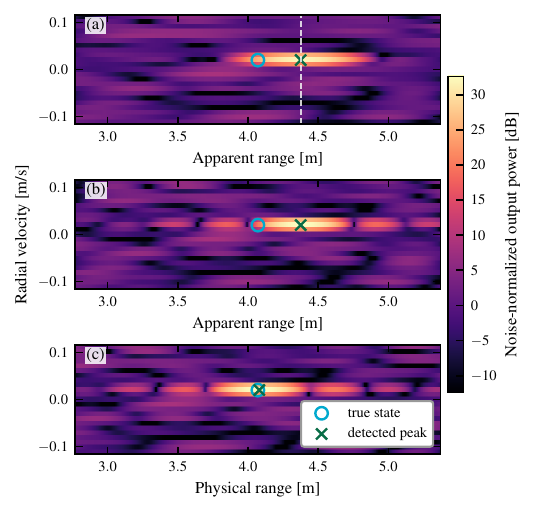}
\caption{Range-Doppler maps of the same CPI on a common decibel scale with \eqref{eq:rdnorm}: (a)~delay transform, (b)~free-space ridge, (c)~wall-aware inverse. All recover the same velocity; the wall-aware map is indexed by physical range, while (a) and (b) are uncorrected and indexed by apparent range.}
\label{fig:rdmaps}
\end{figure}

\begin{figure}[tb]
\centering
\includegraphics[width=0.80\columnwidth]{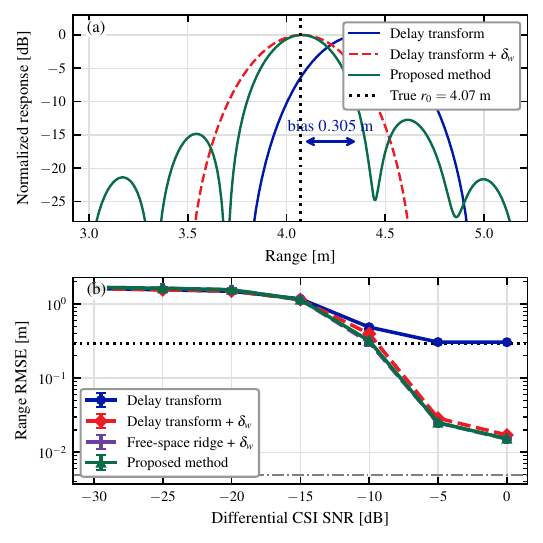}
\caption{Proposed inverse versus delay processing with and without $\delta_w$. (a)~Noiseless range responses on the common grid. (b)~Monte Carlo range RMSE with 95\% bootstrap intervals: uncorrected delay saturates at the wall range bias (dotted: $\Delta R$; dash-dotted: grid step).}
\label{fig:impact}
\end{figure}

\begin{table*}[t]
\caption{Paired Monte Carlo range RMSE (m) and detection probability on the common 5~mm grid. $4000$ trials per SNR through steps (1)--(4), nominal $P_{\rm fa}=0.01$, association gate $0.40$ m. Branches ``$+\,\delta_w$'' are built on $\mathcal{R}+\delta_w$ and labelled by $\mathcal{R}$.}
\label{tab:main}
\centering
\footnotesize
\setlength{\tabcolsep}{5pt}
\renewcommand{\arraystretch}{1.05}
\begin{tabular}{lcccccccccccccc}
\toprule
& \multicolumn{7}{c}{Range RMSE (m)} & \multicolumn{7}{c}{$P_{\rm d}$}\\
\cmidrule(lr){2-8}\cmidrule(lr){9-15}
Branch & $-30$ & $-25$ & $-20$ & $-15$ & $-10$ & $-5$ & $0$ dB
       & $-30$ & $-25$ & $-20$ & $-15$ & $-10$ & $-5$ & $0$ dB\\
\midrule
Delay transform                    & 1.611 & 1.556 & 1.486 & 1.173 & 0.487 & 0.306 & 0.305 & 0.002 & 0.003 & 0.010 & 0.082 & 0.665 & 0.999 & 1.000\\
Delay transform $+\ \delta_w$      & 1.610 & 1.550 & 1.495 & 1.145 & 0.397 & 0.029 & 0.017 & 0.001 & 0.003 & 0.011 & 0.092 & 0.678 & 1.000 & 1.000\\
Free-space ridge                   & 1.658 & 1.630 & 1.560 & 1.174 & 0.434 & 0.306 & 0.305 & 0.002 & 0.004 & 0.011 & 0.116 & 0.756 & 1.000 & 1.000\\
Free-space ridge $+\ \delta_w$     & 1.679 & 1.634 & 1.559 & 1.142 & 0.304 & 0.025 & 0.015 & 0.002 & 0.003 & 0.012 & 0.124 & 0.780 & 1.000 & 1.000\\
Proposed method                    & 1.680 & 1.628 & 1.559 & 1.146 & 0.319 & 0.025 & 0.015 & 0.002 & 0.003 & 0.012 & 0.118 & 0.773 & 1.000 & 1.000\\
\midrule
\multicolumn{15}{l}{\emph{Synthetic $g_w$ in the ridge inverse (residuals of Proposition~\ref{prop:equiv}):}}\\
\ \ linear phase, flat $|g_w|$     & 1.679 & 1.633 & 1.559 & 1.140 & 0.304 & 0.026 & 0.016 & 0.002 & 0.003 & 0.012 & 0.124 & 0.780 & 1.000 & 1.000\\
\ \ exact phase, flat $|g_w|$      & 1.680 & 1.633 & 1.564 & 1.151 & 0.314 & 0.025 & 0.015 & 0.002 & 0.003 & 0.011 & 0.115 & 0.788 & 1.000 & 1.000\\
\ \ exact $|g_w|$, linear phase    & 1.682 & 1.623 & 1.560 & 1.138 & 0.310 & 0.026 & 0.016 & 0.002 & 0.003 & 0.012 & 0.119 & 0.777 & 1.000 & 1.000\\
\bottomrule
\end{tabular}
\end{table*}

\section{Results and Discussions}
\label{sec:results}

We present representative results for the model of Section~II under the configuration of Table~\ref{tab:config}. All branches share the range grid, CSI realization, tapers, gate, and noise model, and every trial realizes steps (1)--(4) of Fig.~\ref{fig:pipeline} in the post-FFT CSI domain. Thresholds are calibrated per processor and per signal-to-noise ratio (SNR) from noise-only trials, and detection requires threshold crossing with correct range association. The differential-CSI SNR is defined on the unit-normalized UE column,
\begin{equation}
\mathrm{SNR}=10\log_{10}\frac{K^{-1}\|\vec{d}_{\rm target}\|_2^2}{\sigma_{\rm rx}^2(1+N_{\rm bg}^{-1})}.
\end{equation}
The transfer matrix, the sampled CP-OFDM identity \eqref{eq:Ykm}, and Proposition~\ref{prop:equiv} were each verified numerically.

Fig.~\ref{fig:rdmaps} shows the RD maps of one CPI at $-10$~dB on the shared decibel scale of \eqref{eq:rdnorm}. Background subtraction removes the wall return and the fixed reflectors, after which the three processors differ mainly in range label, Fig.~\ref{fig:profiles}.

As observed in Fig.~\ref{fig:impact}(a), the noiseless delay transform peaks at $4.375$~m, a range bias of $0.305$~m $\approx\delta_w$, whereas the wall-aware inverse peaks at $4.075$~m, as do both branches rebuilt on $\mathcal{R}+\delta_w$.

Table~\ref{tab:main} and Fig.~\ref{fig:impact}(b) report one paired single-snapshot campaign. The uncorrected branches saturate at their deterministic range bias ($0.305$~m) once the SNR suffices to find the peak, while the corrected branches keep improving. The wall-aware ridge and the free-space ridge on $\mathcal{R}+\delta_w$ agree: their paired RMSE difference is $-1.6$~cm at $-10$~dB and negligible above it, and their $P_{\rm d}$ difference stays within one percentage point, so no detection gain is attributable to the wall model. Both ridge branches outperform the delay transform in the threshold region.

Replacing $g_w$ by synthetic factors that retain only part of it changes the paired differences by at most $1.6$~cm in RMSE and $1.5$ percentage points in $P_{\rm d}$, with varying sign. Perturbing the assumed permittivity or thickness by $\pm10\%$ and $\pm20\%$ shifts the scalar-corrected branches by exactly the change in $\delta_w$ of \eqref{eq:deltaw}, which the wall-aware branch tracks only approximately: it is no more robust than the scalar it implements.

The study is conditional on its idealizations: a single point target treated to first order about a fixed wall background, a receiver that knows the transmitted symbols and is ideally synchronized and phase coherent over the CPI, and unit-column normalization at equal received differential SNR. Finally, $\delta_w(\theta)$ at oblique incidence grows with angle, so a single translation would not serve a distributed aperture.

\section{Conclusions}

We proposed a wall-aware inverse Maxwell range operator for through-wall OFDM sensing, the distorted-Born Jacobian of the discretized Helmholtz equation of the calibrated wall, assembled offline and applied online per snapshot. In one dimension it reduces to the squared wall transmission times a free-space steering kernel, and for constant magnitude and linear phase it is a free-space inverse translated by $\delta_w$. On a simulated $L$-layer wall, with all branches on one grid and a common post-FFT CSI model, the proposed inverse removes the range bias at which wall-unaware processing saturates, and a scalar-corrected free-space inverse tracks it closely, the two differing only marginally in empirical range RMSE and detection probability over the tested sweep; the residual content of the calibrated wall factor changes little even in the threshold region.

\balance
\bibliographystyle{IEEEtran}
\bibliography{references}

@ARTICLE{Nguyen_TGCN_2026,
  author={Nguyen, Nghia Thinh and Do, Tri Nhu},
  journal={IEEE Trans. Green Commun. Netw.}, 
  title={{Generative and Explainable AI for High-Dimensional MIMO-OFDM Channel Estimation in Time--Frequency--Space Domain}}, 
  year={2026},
  volume={10},
  number={},
  pages={3060-3075},
  doi={10.1109/TGCN.2026.3693617},
  ISSN={2473-2400},
  month={},}

@misc{MathWorksWLANSensingDL,
  author       = {{The MathWorks, Inc.}},
  title        = {{Detect Human Presence Using Wireless Sensing with Deep Learning}},
  howpublished = {WLAN Toolbox documentation},
  year         = {2024},
  url          = {https://www.mathworks.com/help/wlan/ug/detect-human-presence-using-wireless-sensing-with-deep-learning.html}
}

@article{LuoTWC2024,
  author={Hongliang Luo and others},
  journal={IEEE Trans. Wireless Commun.},
  title={{Integrated Sensing and Communications in Clutter Environment}},
  year={2024},
  volume={23},
  number={9},
  pages={10941--10956},
  doi={10.1109/TWC.2024.3377184}
}

@article{AhmadTAES2005,
  author  = {Fauzia Ahmad and Moeness G. Amin and Saleem A. Kassam},
  journal = {IEEE Trans. Aerosp. Electron. Syst.},
  title   = {{Synthetic Aperture Beamformer for Imaging Through a Dielectric Wall}},
  year    = {2005},
  volume  = {41},
  number  = {1},
  pages   = {271--283},
  doi     = {10.1109/TAES.2005.1413768}
}

@article{DehmollaianTGRS2008,
  author  = {Mojtaba Dehmollaian and Kamal Sarabandi},
  journal = {IEEE Trans. Geosci. Remote Sens.},
  title   = {{Refocusing Through Building Walls Using Synthetic Aperture Radar}},
  year    = {2008},
  volume  = {46},
  number  = {6},
  pages   = {1589--1599},
  doi     = {10.1109/TGRS.2008.916212}
}

@BOOK{Chew1995Waves,
  author    = {Weng Cho Chew},
  title     = {Waves and Fields in Inhomogeneous Media},
  publisher = {IEEE Press},
  address   = {Piscataway, NJ, USA},
  year      = {1995}
}

@article{ChewTMI1990,
  author  = {Weng Cho Chew and Y. M. Wang},
  journal = {IEEE Trans. Med. Imag.},
  title   = {{Reconstruction of Two-Dimensional Permittivity Distribution Using the Distorted {Born} Iterative Method}},
  year    = {1990},
  volume  = {9},
  number  = {2},
  pages   = {218--225},
  doi     = {10.1109/42.56334}
}

@article{LawJCP2026,
  author  = {Yann-Meing Law},
  journal = {J. Comput. Phys.},
  title   = {{The High-Order Hermite Discrete Correction Function Method for Surface-Driven Electromagnetic Problems}},
  year    = {2026},
  volume  = {563},
  pages   = {115058},
  doi     = {10.1016/j.jcp.2026.115058}
}

@article{DaiCOMST2026,
  author  = {Qianglong Dai and others},
  title   = {A Tutorial on {MIMO-OFDM} {ISAC}: From Far-Field to Near-Field},
  journal = {IEEE Commun. Surveys Tuts.},
  volume  = {28},
  pages   = {4319--4358},
  year    = {2026},
  doi     = {10.1109/COMST.2025.3650568}
}

@article{LiaoTAES2025,
  author  = {Jiancheng Liao and Xiaolu Zeng and Xiaopeng Yang and Zixiang Yin and Zihan Chen},
  title   = {{MIMO} Through-the-Wall Radar 3-{D} Imaging Using Propagation Compensation and Coherence Factor},
  journal = {IEEE Trans. Aerosp. Electron. Syst.},
  volume  = {61},
  number  = {6},
  pages   = {17213--17226},
  year    = {2025},
  doi     = {10.1109/TAES.2025.3602758}
}

@article{YaoTIM2025,
  author  = {Yu Yao and others},
  title   = {Target Localization and Wall Parameters Estimation via Distributed Through-Wall Imaging Radar},
  journal = {IEEE Trans. Instrum. Meas.},
  volume  = {74},
  pages   = {1--11},
  year    = {2025},
  doi     = {10.1109/TIM.2024.3502875}
}

\end{document}